\documentclass[final,1p,times,intlimits]{article}

\usepackage{mathtools}
\usepackage{amssymb}
\usepackage{amsthm}
\usepackage{xcolor}
\usepackage{tikz}

\DeclareSymbolFont{bbold}{U}{bbold}{m}{n}
\DeclareSymbolFontAlphabet{\mathbbold}{bbold}

\newcommand{\N}{\mathbb{N}}
\newcommand{\Z}{\mathbb{Z}}

\newcommand{\R}{\mathbb{R}}
\newcommand{\C}{\mathbb{C}}
\newcommand{\K}{\mathbb{K}}
\newcommand{\1}{\mathbbold{1}}

\newcommand{\loc}{\mathrm{loc}}
\newcommand{\lu}{\mathrm{unif}}

\newcommand{\M}{\mathcal{M}_{\loc,\lu}(\R)}

\newcommand{\Dir}{{\rm D}}
\newcommand{\Neu}{{\rm N}}
\newcommand{\SL}{\mathrm{SL}}

\newcommand{\s}{{\rm s}}
\renewcommand{\c}{{\rm c}}
\renewcommand{\sc}{{\rm sc}}
\newcommand{\ac}{{\rm ac}}
\newcommand{\pp}{{\rm pp}}
\newcommand{\disc}{{\rm disc}}
\newcommand{\ess}{{\rm ess}}

\newcommand{\dist}{\mathrm{dist}}

\renewcommand{\P}{\mathbb{P}}

\makeatletter
\newcommand{\uD}{\u@DN{\Dir}}
\newcommand{\uN}{\u@DN{\Neu}}
\newcommand{\u@DN}[1]{u_{#1\mkern-1mu}}
\makeatother

\DeclareMathOperator{\spt}{spt}

\renewcommand{\Im}{\operatorname{Im}}
\DeclareMathOperator{\lin}{lin}

\providecommand{\form}{\tau}
\providecommand{\scpr}[2]{\left( #1 \,\middle|\, #2 \right)}
\providecommand{\dupa}[2]{\left\langle #1 , #2 \right\rangle}
\renewcommand{\sp}{\scpr}
\newcommand{\from}{\colon}

\let\phi\varphi

\let\leq\leqslant

\let\geq\geqslant

\makeatletter
\def\@row#1,{#1\@ifnextchar;{\@gobble}{&\@row}}
\def\@matrix{%
    \expandafter\@row\my@arg,;%
    \@ifnextchar({\\ \get@in@paren{\@matrix}}{\after@matrix}%
    }
\def\matrixtype#1#2#3{%
    \ifmmode\def\after@matrix{\end{#2}\right#3}%
    \else\def\after@matrix{\end{#2}\right#3$}$\fi\iffalse$\fi
    \left#1\begin{#2}\get@in@paren{\@matrix}%
    }
\def\@column#1,{#1\@ifnextchar;{\@gobble}{\\ \@column}}
\newcommand\vect{}
\def\svect(#1){\left(\begin{smallmatrix}\@column#1,;\end{smallmatrix}\right)}
\def\vect{\get@in@paren{\@vect}}
\def\@vect{\left(\begin{matrix}\expandafter\@column\my@arg,;\end{matrix}\right)}
\def\get@in@paren#1({\def\my@arg{}\def\my@rest{}\def\after@get{#1}\get@arg}
\let\e@a\expandafter
\def\get@arg#1){\e@a\kl@test\my@rest#1(;}
\def\kl@test#1(#2;{\e@a\def\e@a\my@arg\e@a{\my@arg#1}%
                   \ifx:#2:\let\my@exec\after@get
                   \else\let\my@exec\get@arg
                        \e@a\def\e@a\my@arg\e@a{\my@arg(}%
                        \def@rest#2;%
                   \fi\my@exec}
\def\def@rest#1(;{\def\my@rest{#1\kl@zu}}
\def\kl@zu{)}

\makeatletter
\newcommand\MyPairedDelimiter{%
  \@ifstar{\My@Paired@Delimiter{{}}}
          {\My@Paired@Delimiter{}}%
}
\newcommand\My@Paired@Delimiter[4]{%
  \newcommand#2{%
    \@ifstar{\start@PD{#1}{\delimitershortfall=-1sp}{#3}{#4}}
            {\start@PD{#1}{}{#3}{#4}}%
  }%
}
\newcommand\start@PD[5]{%
  #1\mathopen{\mathpalette\put@delim@helper{\put@delim{#2}{#3}{.}{#5}}}%
  #5%
  \mathclose{\mathpalette\put@delim@helper{\put@delim{#2}{.}{#4}{#5}}}%
}
\newcommand\put@delim@helper[2]{%
  \hbox{$\m@th\nulldelimiterspace=0pt #2#1$}%
}
\newcommand\put@delim[5]{%
  \setbox\z@\hbox{$\m@th#5{#4}$}%
  \setbox\tw@\null
  \ht\tw@\ht\z@ \dp\tw@\dp\z@
  #1#5%
  \left#2\box\tw@\right#3%
}

\makeatother
\MyPairedDelimiter*{\abs}{\lvert}{\rvert}
\MyPairedDelimiter*{\norm}{\lVert}{\rVert}
\MyPairedDelimiter{\set}{\{}{\}}

\newcommand\rlim{
\mathchoice{\vcenter{\hbox{${\scriptstyle{+}}$}}}
{\vcenter{\hbox{$\scriptstyle{+}$}}}
{\vcenter{\hbox{$\scriptscriptstyle{+}$}}}
{\vcenter{\hbox{$\scriptscriptstyle{+}$}}}}

\theoremstyle{plain} % default
\newtheorem{theorem}{Theorem}[section]

\newtheorem{lemma}[theorem]{Lemma}
\newtheorem{proposition}[theorem]{Proposition}
\theoremstyle{definition}

\newtheorem*{definition}{Definition}
\newtheorem{remark}[theorem]{Remark}

\usepackage{enumitem}

\setenumerate[1]{nolistsep} % Only the level 1
\setenumerate[2]{nolistsep} % Only the level 2

\newcommand{\Hmm}[1]{\leavevmode{\marginpar{\tiny%
$\hbox to 0mm{\hspace*{-0.5mm}$\leftarrow$\hss}%
\vcenter{\vrule depth 0.1mm height 0.1mm width \the\marginparwidth}%
\hbox to 0mm{\hss$\rightarrow$\hspace*{-0.5mm}}$\\\relax\raggedright #1}}}

\begin{document}

\medmuskip=4mu plus 2mu minus 3mu
\thickmuskip=5mu plus 3mu minus 1mu
\belowdisplayshortskip=9pt plus 3pt minus 5pt

\title{Zero measure Cantor spectra for continuum one-dimensional quasicrystals}

\author{Daniel Lenz, Christian Seifert and Peter Stollmann}

\date{}

\maketitle

\begin{abstract}
  We study Schr\"odinger operators on $\R$ with measures as potentials.
  Choosing a suitable subset of measures we can work with a dynamical system consisting of measures. We then relate properties of this dynamical system  with spectral properties of the associated operators.
  The constant spectrum in the strictly ergodic case coincides with the union of the zeros of the Lyapunov exponent and the set of non-uniformities of the transfer matrices.
  This result enables us to prove Cantor spectra of zero Lebesgue measure for a large class of operator families, including  many operator families generated by aperiodic subshifts.
  
  \bigskip
  
  \medskip


\end{abstract}

\section{Introduction}

Mathematical models for quasicrystals have been studied intensively during the last decades both from the point of view of mathematics and of physics, 
see e.g. the survey articles \cite{Suto95,Damanik00,DamanikEmbreeGorodetski12}. 
In fact,
  right from the beginning in the heuristic investigations \cite{KohmotoKadanoffTang83,Ostlundetal83}  and  in the  rigorous studies   \cite{Casdagli86,BellisardIochumScoppolaTestard1989,Suto87} the  remarkable spectral properties of discrete one-dimensional models has been a key focus of research. In the early days, these spectral phenomena seemed rather strange, especially from the point of view of classical Schr\"odinger operators; with the development of random operator theory it turned out that strange spectral properties are generic in certain ways.

%In fact,
%  right from the beginning in the heuristic investigations \cite{KohmotoKadanoffTang83,Ostlundetal83}  and  in the  rigorous studies   \cite{Casdagli86,BellisardIochumScoppolaTestard1989,Suto87} the  remarkable and earlier on considered ``strange'' spectral behavior of discrete one-dimensional models has been a key focus of research.

   More specifically, the Hamiltonians for quasicrystals tend to have
  Cantor sets of Lebesgue measure zero  as spectra and their spectral type tends to be purely singular continuous. So far, these topics have been thoroughly investigated for various examples in the discrete case, i.e.\ for operators  on $\Z$. However, from the point of view of modeling  there is no reason to restrict attention to discrete models. It is rather quite natural to consider continuum models on $\R$ as well. In fact,  absence of eigenvalues due to so called Gordon arguments, and absence of absolutely continuous spectrum due to Kotani-Remling theory have been investigated for various classes of  models in the continuum over the last ten  years or so \cite{DamanikStolz2000,Seifert2011,SeifertVogt2013,Klassert07,KlassertLenzStollmann2011}.
 Also, generic singular continuous spectrum for certain classes of models on Delone sets (even in arbitrary dimension) has been established in \cite{LenzStollmann06}.  In this sense, it seems fair to say that the basic questions concerning spectral type for continuum models in one dimension are well understood. In this article, it is now our aim to treat the question of Cantor spectrum of Lebesgue measure zero.  We will do so in a rather general setting. More specifically, we will deal with continuum one-dimensional Hamiltonians of the form
$$-\Delta+\mu\;\: \mbox{in}\;\: L_2(\R),$$
 where $\mu$ is a suitable (local) measure on $\R$. In this way, we will allow
 for very general potentials.

The basic strategy to prove Cantor spectrum of Lebesgue measure zero  follows the method developed  in \cite{Lenz2002} in the discrete setting.  However, due to the general nature of the allowed potentials there is quite some work to be done and  many details to be taken care of. In fact, our work can be seen as a twofold generalization from the discrete case: first from discrete Schr\"odinger operators to continuum ones and then from continuum Schr\"odinger operators to those with a measure as a potential.  The first step is to link ergodic features of subsets $\Omega$ of potentials (i.e., measures) with spectral properties of the associated Hamiltonians.
In case $\Omega$ is strictly ergodic with respect to translations we observe constancy of the spectrum.
This spectrum will then be characterized with the help of Lyapunov exponents and uniformity of transfer matrices.
Employing geometric properties of the potentials as in \cite{KlassertLenzStollmann2011} we can exclude absolutely continuous spectrum, which by Kotani theory yields results on the Lyapunov exponents.
In this way we can prove Cantor spectra of zero measure for a large class of operator families.

\medskip

The paper is organized as follows. In Section \ref{sec:Schr} we introduce Schr\"odinger operators on $\R$ with measures as potentials. Section \ref{sec:Spec}
is devoted to general statements on the spectrum of random Schr\"odinger operators. The transfer matrices and the Lyapunov exponent are introduced in Section \ref{sec:Transfer}.
In Section \ref{sec:CharSpec} we characterize the spectrum of the family of operators by means of the zeros of the Lyapunov exponent and the set of energies where the transfer matrices are not uniform.
We focus on the absolutely continuous spectrum in Section \ref{sec:CharacSpec} and explain the Ishii-Pastur-Kotani Theorem for our setting.
Conditions on absence of absolutely continuous spectrum are considered in Section \ref{sec:AbsacSpec}.
Finally, in Section \ref{sec:Cantor} we prove Cantor spectra of zero measure for suitable operator families.
We also provide a scheme to produce examples generated by aperiodic subshifts over finite alphabets.
In the appendix we state and prove a semi-uniform ergodic theorem for continuous-time subadditive processes.

\smallskip

Parts of this paper are based on the 2012  PhD thesis of one of the authors \cite{Seifert2012}. 
After conceiving this paper we learned about the recent work  \cite{DamanikFillmanGorodetski2012} containing results similar to ours in the situation of 'nice' potentials.

\section{Schr\"odinger operators with measures}
\label{sec:Schr}

In this section we introduce Schr\"odinger operators with (signed) measures as potentials.

We say that a map $\mu$ from the bounded Borel sets in $\R$ to $\R$
%\[\mu\from\set{B\subseteq \R;\; B \;\text{is a bounded Borel set}}\to \R\]
is a \emph{local measure} if $\1_K\mu:=\mu(\cdot\cap K)$ is a signed Radon measure for all compact sets $K\subseteq \R$.
Then it is easy to see that there exist a unique nonnegative Radon measure $\nu$ on $\R$ and a measurable function $\sigma\from\R\to\R$ such that
$\abs{\sigma} = 1$ $\nu$-a.e.~and $\1_K\mu = \1_K\sigma\nu$ for all compact sets $K\subseteq\R$. We call $\nu$ the total variation of $\mu$ and write $\abs{\mu}:= \nu$.
A local measure $\mu$ is called \emph{uniformly locally bounded} or \emph{translation bounded} if
\[\norm{\mu}_\lu := \sup_{t\in\R} \abs{\mu}((t,t+1]) < \infty.\]
Let $\M$ denote that space of all uniformly locally bounded local measures. For $\mu\in\M$ and an interval $[a,b]\subseteq\R$ we observe
$\abs{\mu}([a,b])\leq \norm{\mu}_\lu (b-a + 1)$.

Let
\begin{align*}
  D(\form_0) & := W_2^1(\R),\quad
  \form_0(u,v) := \int u' \overline{v}'
\end{align*}
be the classical Dirichlet form associated with the Laplacian $-\Delta$ in $L_2(\R)$. Let $\mu\in\M$. 
Then $\mu$ defines an infinitesimally form small perturbation with respect to $\form_0$; cf.~\cite[Lemma 1.1.1]{Seifert2012}. 
Indeed, by Sobolev's lemma, for any $\gamma>0$ there exists $C_\gamma>0$ such that for any interval $I\subseteq \R$ of length $1$ we obtain
\[\norm{u|_I}_\infty^2 \leq \gamma \norm{u'|_I}_2^2 + C_\gamma \norm{u|_I}_2^2 \quad(u\in W_2^1(\R)).\]
Hence, for $u\in W_2^1(\R)$, we obtain
\begin{align*}
  \abs{\mu(u,u)} & = \int_\R \abs{u(t)}^2\,d\mu(t) = \sum_{k\in\Z} \int_{(k,k+1]} \abs{u(t)}^2\,d\abs{\mu}(t) \\
  & \leq \sum_{k\in\Z} \norm{u|_{[k,k+1]}}_\infty^2 \norm{\mu}_\lu \leq \norm{\mu}_\lu \gamma \form_0(u,u) + \norm{\mu}_\lu C_\gamma \norm{u}_2^2.
\end{align*}
Thus, the form $\form_\mu:= \form_0+\mu$ defined by
\begin{align*}
  D(\form_\mu) & := W_2^1(\R),\quad
  \form_\mu(u,v) := \int u' \overline{v}' + \int u \overline{v}\, d\mu
\end{align*}
is densely defined, semibounded from below, symmetric and closed. Let $H_\mu$ be the unique self-adjoint operator in $L_2(\R)$ associated with $\form_\mu$, i.e.
\[\sp{H_\mu u}{v} = \form_\mu(u,v) \quad(u\in D(H_\mu), v\in D(\form_\mu)),\]
and $D(H_\mu)$ is dense in $D(\form_\mu)$ equipped with the form norm.

\begin{remark} For  $\mu\in\M$ we have defined the operator $H_\mu$ via forms. In fact, various definitions can be found in the literature. This is shortly discussed in the present remark: 
  \begin{enumerate}
    \item
      Let $H_{\mu}^\dist$ be the maximal operator associated with $-\Delta+\mu$ (in the distributional sense), viz
      \begin{align*}
	D(H_\mu^\dist) & := \set{u\in L_2(\R)\cap C(\R);\; -u'' + u\mu\in L_2(\R)},\\
	H_\mu^\dist u & := -u'' + u\mu.
      \end{align*}
      This means  for $u,f\in L_2(\R)$ one has $u\in D(H_\mu^\dist)$, $H_\mu^\dist u = f$ if and only if
      \[
	\int_\R f\phi = -\!\int_\R u\phi'' + \int_\R u\phi\,d\mu \qquad \bigl(\phi\in C_\c^\infty(\R)\bigr).
      \]
    \item
      We can also define a realization of the operator along the lines of Sturm-Liouville theory, see \cite{BenAmorRemling2005}.
      For $u\in W_{1,\loc}^1(\R)$ (choosing always the continuous representative) one defines $A_\mu u\in L_{1,\loc}(\R)$ by
      \[(A_\mu u)(t) := u'(t) - \int_0^t u(s)\, d\mu(s)\]
      for a.a.~$t\in\R$.
      Here,
      \[\int_0^t\ldots \, d\mu = \begin{cases}
			    \int_{(0,t]} \ldots\, d\mu & ,\ t\geq 0,\\
			    -\int_{(t,0]} \ldots\, d\mu & ,\ t<0.
                          \end{cases}\]
      Then we define $H_{\mu}^\SL$ as
      \begin{align*}
	D(H_{\mu}^\SL) & := \set{u\in L_2(\R)\cap W_{1,\loc}^1(\R);\;
				A_\mu u \in W_{1,\loc}^1(\R),\ (A_\mu u)'\in L_2(\R)}, \\
	H_{\mu}^\SL u  & := -(A_\mu u)'.
      \end{align*}
  \end{enumerate}    
      It turns out that all these operators agree. In fact, 
      as shown in \cite[Theorem 3.6]{SeifertVogt2013} we have
      \[H_\mu = H_\mu^\dist = H_\mu^\SL.\]
  \end{remark}

\section{Spectra of random Schr\"odinger Operators}
\label{sec:Spec}

We first show that the vague topology $\sigma(C_c(\R)',C_c(\R))$ on bounded subsets of $\M$ is compact and metrizable. This is, of course, well known. In the context of diffraction of quasicrystals a thorough study can be found in \cite{BaakeLenz04,BaakeLenz2005}.

\begin{proposition}
  Let $\Omega\subseteq \M$ be $\norm{\cdot}_\lu$-bounded and closed with respect to the vague topology.
  Then $\Omega$ is  vaguely compact. 
  Furthermore,
  the vague topology on $\Omega$ is induced by some metric, i.e., $\Omega$ is metrizable.
\end{proposition}

\begin{proof}
	We first note that $C_c(\R)$ is the inductive limit of
$(C_0(-n,n))_{n\in\N}$. Thus, $C_c(\R)$ can be equipped with the inductive
topology. Since $C_0(-n,n)$ is separable for all $n\in\N$, also $C_c(\R)$ is
separable.

For $A>0$ let $U_A:=\set{f\in C_c(\R);\; \abs{f(t)}\leq
Ae^{-\abs{t}}\quad(t\in\R)}$. Then $U_A$ is a neighborhood of $0\in C_c(\R)$.
There exists $C\geq 0$ such that $\norm{\mu}_\lu \leq C$ for all
$\mu\in\Omega$. An easy computation yields
\[\abs{\dupa{f}{\mu}} \leq CA \frac{2e}{e-1} \quad(\mu\in\Omega, f\in U_A).\]

So, for $A:= \frac{e-1}{2eC}$ we have $\Omega\subseteq U_A^\circ$, where
\[U^\circ := \set{\mu\in C_c(\R)';\; \abs{\dupa{f}{\mu}}\leq 1\quad(f\in U)}\]
is the absolute polar set of $U$. The Theorem of Alaoglu-Bourbaki (\cite[Satz 23.5]{MeiseVogt1992}) assures that
$U_A^\circ$ is $\sigma(C_c(\R)',C_c(\R))$-compact. Since $\Omega$ is a closed
subset of $U_A^\circ$, $\Omega$ is compact, too.

In order to show that $\Omega$ is metrizable note that the initial topology
$\mathcal{T}$ on $C_c(\R)'$ corresponding to a countable dense set of $C_c(\R)$
 is semimetrizable (i.e.\ it comes from some semimetric) and also separated, and hence induced by some metric. Since
the identity
\[I\from (\Omega,\sigma(C_c(\R)',C_c(\R)))\to (\Omega,\mathcal{T})\]
is bijective, continuous and maps a compact onto a separated space it is in
fact a homeomorphism.
\end{proof}

From now on we assume that $\Omega\subseteq \M$ is $\norm{\cdot}_\lu$-bounded and
 closed with respect to the vague topology. In this setting we always equip
$\Omega$ with the vague topology
such that $\Omega$ becomes a compact metric space.
Furthermore, assume $\Omega$ to be \emph{translation invariant},
i.e., for $\omega\in \Omega$ let also $\omega(\cdot+t)\in \Omega$ for all $t\in\R$.
Then, the additive group $\R$ induces a group action of translations on $\Omega$
 via $\alpha\from
\R\times\Omega\to \Omega$, $\alpha_t(\omega):=\alpha(t,\omega):=
\omega(\cdot+t)$.

\begin{lemma}
$\alpha$ is continuous.
\end{lemma}

\begin{proof}
Let $(t_n)$ in $\R$, $t_n\to t$ and $(\omega_n)$ in $\Omega$, $\omega_n\to \omega$.
Then
\[\alpha_t(\omega_n) = \omega_n(\cdot+t)\to \omega(\cdot+t) = \alpha_t(\omega).\]
Let $f\in C_c(\R)$. There exists $R>0$ such that $\spt f\subseteq [-R,R]$. Let $\varepsilon>0$. There exists $N\in \N$ such that
\[\abs{\int f\,d(\alpha_t(\omega_n)-\alpha_t(\omega))}\leq \varepsilon \quad(n\geq N).\]
Furthermore, by uniform continuity of $f$ and convergence of $(t_n)$, there exists $N'\geq N$ such that $\abs{t_n-t}\leq 1$ and
\[\abs{f(\cdot-t_n) - f(\cdot-t)}\leq \varepsilon\]
for $n\geq N'$.
Hence, for $n\geq N'$, we obtain
\begin{align*}
& \abs{\int f\,d(\alpha_{t_n}(\omega_n) - \alpha_t(\omega))} \\
& \leq \int_{[-R+t-1,R+t+1]} \abs{f(\cdot-t_n)-f(\cdot-t)}\, d\abs{\omega_n} + \abs{\int f(\cdot-t)\, d(\omega_n-\omega)} \\
& \leq \varepsilon \abs{\omega_n}([-R+t-1,R+t+1]) + \varepsilon = \left(\norm{\omega_n}_\lu(2R+3) + 1\right)\varepsilon.
\end{align*}
As $\Omega$ is $\norm{\cdot}_\lu$-bounded we conclude $\alpha_{t_n}(\omega_n)\to \alpha_t(\omega)$.
\end{proof}

\begin{remark} The previous results give that $(\Omega, \alpha)$ is a topological dynamical system. In the context of diffraction on locally compact abelian groups such systems
were introduced and studied under the name of translation bounded measures dynamical systems (TMDS) in \cite{BaakeLenz04,BaakeLenz2005}.
\end{remark}

For $\omega\in\Omega$ the operator $H_\omega$ can be defined
as above by means of the form
\begin{align*}
 D(\form_\omega) & := W_2^1(\R),\quad  \form_\omega(u,v) := \form_0(u,v) + \int u\overline{v}\, d\omega.
\end{align*}
It is easy to see that the lower bound of $H_\omega$ depends on $\norm{\omega}_\lu$; see also \cite[Lemma 1.1]{KlassertLenzStollmann2011}.
Since $\Omega$ is $\norm{\cdot}_\lu$-bounded there
exists $\gamma\in\R$ such that $H_\omega\geq -\gamma$ for all $\omega\in\Omega$.

We will now establish continuity of the mapping $\omega\mapsto H_\omega$ in strong resolvent sense. In the special case of random operators on Delone sets, this can be found in e.g. \cite{LenzStollmann06} (see \cite{Klassert07} as well). The general case seems not to be available in the literature,  we therefore include full proofs. We need some preparation.

\begin{lemma}
\label{lem:fourier}
Let $\nu$ be a finite signed Radon measure on $\R$. Then $\nu\in W_2^1(\R)'$ and
\[\norm{\nu}_{W_2^1(\R)'} \leq \norm{\widehat{J(-(\cdot))}\cdot \hat{\nu}}_{L_2(\R)},\]
where $\hat{J}(p) = \frac{1}{\sqrt{1+p^2}}$ ($p\in\R$) and the hat indicates the Fourier transform.
\end{lemma}

\begin{proof}
  There exists a unique $J\in L_2(\R)$ such that $\hat{J}(p) = \frac{1}{\sqrt{1+p^2}}$ ($p\in\R$) and
  \[J * f = \sqrt{2\pi} (-\Delta+1)^{-1/2}f \quad(f\in L_2(\R)).\]
  Let $v\in C_c^\infty(\R)$. Then
  \begin{align*}
  \int \int \abs{J(x-y)}\abs{v(y)}\, dy \, d\abs{\nu}(x) & \leq \int \norm{J(x-\cdot)}_{L_2(\R)} \norm{v}_{L_2(\R)} \, d\abs{\nu}(x) \\
  & = \norm{J}_{L_2(\R)} \norm{v}_{L_2(\R)} \abs{\nu}(\R) < \infty.
  \end{align*}
  Hence, Fubini's Theorem applies and we obtain
  \begin{align*}
  \abs{\int v\, d\nu} & = \abs{\int (-\Delta+1)^{-1/2}(-\Delta+1)^{1/2}v\, d\nu} \\
  & = \frac{1}{\sqrt{2\pi}}\abs{\int \int J(x-y) (-\Delta+1)^{1/2}v(y)\, dy\, d\nu(x)} \\
  & = \frac{1}{\sqrt{2\pi}}\abs{\int \Bigl(\int J(x-y)\, d\nu(x)\Bigr) (-\Delta+1)^{1/2}v(y)\, dy} \\
  & \leq \frac{1}{\sqrt{2\pi}}\norm{J(-(\cdot))*\nu}_{L_2(\R)} \norm{(-\Delta+1)^{1/2}v}_{L_2(\R)} \\
  & = \norm{\widehat{J(-(\cdot))}\cdot \hat{\nu}}_{L_2(\R)} \norm{v}_{W_2^1(\R)}.
  \end{align*}	
  By density of $C_c^\infty(\R)$ in $W_2^1(\R)$ the assertion follows.
\end{proof}

\begin{lemma}[{\cite[Lemma 1]{BrascheFigariTeta1998}}]
\label{lem:unif_bound_on_fourier}
Let $\nu,\nu_k$ be finite signed Radon measures on $\R$ ($k\in\N$), $\nu_k\to \nu$ weakly. Then $\sup_{k\in\N} \norm{\hat{\nu}_k}_\infty<\infty$.
\end{lemma}

\begin{proof}
  Weak convergence of $(\nu_k)$ is exactly pointwise convergence of the corresponding linear functionals on $C_b(\R)$.
  The uniform boundedness principle yields $\sup_{k\in\N} \norm{\nu_k}_{C_b(\R)'}
  <\infty$.
  Furthermore, for $k\in\N$ and $t\in\R$ we have
  \[\abs{\hat{\nu}_k(t)} = \abs{\frac{1}{\sqrt{2\pi}} \nu_k(e^{-it(\cdot)})}
  %\leq \frac{1}{\sqrt{2\pi}} \int_\R \abs{e^{-its}} \, d\abs{\nu_k}(s)
  \leq \frac{1}{\sqrt{2\pi}}  \sup_{k\in\N} \norm{\nu_k}_{C_b(\R)'}.\]
  Hence,
  \[\sup_{k\in\N} \norm{\hat{\nu}_k}_\infty<\infty. \qedhere \]
\end{proof}

\begin{theorem}
\label{thm:src}
  Let $(\mu_n)$ in $\M$ be bounded, $\mu \in \M$, $\mu_n\to \mu$ vaguely. Then $H_{\mu_n}\to H_\mu$ in strong resolvent sense.
\end{theorem}

\begin{proof}
  Let $u\in C_c^\infty(\R)$ and $\nu_k:= u\mu_k$ ($k\in\N$), $\nu:= u\mu$. Then $\nu_k, \nu$ are finite signed
  Radon measures on $\R$ and $\nu_k\to \nu$ weakly.

  % and $\sup_{k\in\N} \norm{\hat{\nu}_k}_\infty < \infty$ by Lemma
  %\ref{lem:unif_bound_on_fourier}.
  Thus $\hat{\nu}_k(p)\to \hat{\nu}(p)$ for all $p\in \R$. Furthermore,
  \[\sup_{k\in\N} \norm{\hat{\nu}_k-\hat{\nu}}_\infty < \infty\]
  by Lemma \ref{lem:unif_bound_on_fourier}.
  We conclude that
  \[\norm{\widehat{J(-(\cdot))}\cdot (\hat{\nu}_k-\hat{\nu})}_{L_2(\R)}\to 0\]
  by Lebesgue's dominated convergence theorem, where $J$ is as in Lemma
\ref{lem:fourier}. In view of Lemma \ref{lem:fourier},
  \[\norm{\form_{\mu_k}(u,\cdot) - \form_{\mu}(u,\cdot)} = \norm{\nu_k-\nu}_{W_2^1(\R)'} \leq \norm{\widehat{J(-(\cdot))}\cdot (\hat{\nu}_k-\hat{\nu})}_{L_2(\R)} \to 0.\]
  Now, \cite[Theorem A.1]{StollmannVoigt1996} yields the assertion.
\end{proof}

We say that $(\Omega,\alpha)$ is \emph{ergodic} with ergodic measure $\P$ if every $\alpha$-invariant measurable subset
$A\subseteq \Omega$ satisfies $\P(A)\in\set{0,1}$. If the ergodic measure $\P$
is unique then $(\Omega,\alpha,\P)$ is \emph{uniquely ergodic}.
Furthermore, $(\Omega,\alpha)$ is called \emph{minimal} if every orbit $\mathcal{O}(\omega):=\set{\alpha_t(\omega);\; t\in\R}$ is dense in $\Omega$.
If $(\Omega,\alpha,\P)$ is uniquely ergodic and minimal, then we call it \emph{strictly ergodic}.

Note that if $(\Omega,\alpha,\P)$ is ergodic, then
$$H_{\alpha_t(\omega)} =
U(-t)H_\omega U(t),$$ where $U$ is the shift group on $L_2(\R)$, i.e. $U(t)f =
f(\cdot-t)$. This property of the family is also sometimes known as \emph{covariance}. It means that  the operator family $(H_\omega)_{\omega\in\Omega}$ is what is known as an  \emph{ergodic family}. Thus, we obtain $\P$-almost sure constancy of spectral information in the following sense.

\begin{proposition}[{\cite[Proposition V.2.4]{CarmonaLacroix1990}}]
\label{prop:ac_const}
Let $\Omega\subseteq \M$ be $\norm{\cdot}_\lu$-bounded, vaguely closed and translation invariant,
$\alpha$ the group action of $\R$ on $\Omega$.
Let $(\Omega,\alpha,\P)$ be ergodic. Then there exists closed subsets $\Sigma,\Sigma_{\bullet}\subseteq \R$ such that for $\P$-a.a.~$\omega\in\Omega$ we have
\[\sigma(H_\omega) = \Sigma,\quad \sigma_\bullet(H_\omega) = \Sigma_\bullet
\quad(\bullet\in\set{\s,\c,\ac,\sc,\pp,\disc}).\]
\end{proposition}

Proposition \ref{prop:ac_const} tells us nothing about some particular operator
 $H_\omega$ of the operator family.
However, in case $(\Omega,\alpha)$ is minimal we can obtain constancy of the spectrum as a set. This type of result is well known in the discrete case (see e.g. \cite{Lenz2002,LenzStollmann03} for corresponding results on arbitrary dimensional Delone sets).

\begin{theorem}
Let $\Omega\subseteq \M$ be $\norm{\cdot}_\lu$-bounded, vaguely closed and translation invariant,
$\alpha$ the group action of $\R$ on $\Omega$.
Let $(\Omega,\alpha)$ be minimal. Then there exists $\Sigma\subseteq \R$ such
that
\[\sigma(H_\omega) = \Sigma \quad(\omega\in \Omega).\]
\end{theorem}

\begin{proof}
Let $\omega,\omega'\in\Omega$. If $\omega$ and
$\omega'$ are on the same orbit, i.e.~$\mathcal{O}(\omega) =
\mathcal{O}(\omega')$, we obtain $\sigma(H_\omega) =
\sigma(H_{\omega'})$ by unitary equivalence of the corresponding operators.
Otherwise, by minimality, there exists
$(\omega_k)$ in $\mathcal{O}(\omega)$ such that $\omega_k\to \omega'$. Then
 $\sigma(H_{\omega_k}) = \sigma(H_\omega)$ for all $k\in\N$.
By Theorem \ref{thm:src} we have $H_{\omega_k}\to H_{\omega'}$ in strong resolvent sense.
By \cite[Theorem VIII.24]{ReedSimon1} for $E\in \sigma(H_{\omega'})$
there is $E_k \in \sigma(H_{\omega_k})$ ($k\in\N$) satisfying $E_k\to E$.
But $\sigma(H_{\omega_k}) = \sigma(H_{\omega})$ for all $k\in\N$ and
$\sigma(H_{\omega})$ is closed, so $E \in \sigma(H_{\omega})$. Thus,
we have shown $\sigma(H_{\omega'}) \subseteq \sigma(H_{\omega})$. Interchanging
 the roles of $\omega$ and $\omega'$ yields $\sigma(H_{\omega}) \subseteq
\sigma(H_{\omega'})$ and therefore $\sigma(H_{\omega}) =
\sigma(H_{\omega'})$.
\end{proof}

\begin{remark}
A remarkable result due to Last and Simon shows that  minimality implies  even constancy of the absolutely continuous spectrum \cite{LastSimon1999}. The argument is given there for discrete operators as well as for continuum Schr\"odinger operators with potentials which are continuous functions on the dynamical system.
  By \cite{JitomirskayaSimon1994} the singular continuous and the pure point spectra need not be constant.
\end{remark}

\begin{lemma}
\label{lem:disc_spectrum}
Let $(\Omega,\alpha,\P)$ be ergodic.
Then $\Sigma_\disc = \varnothing$. Hence, $\Sigma$ does not contain isolated points.
\end{lemma}

\begin{proof}
  The first assertion follows the lines of \cite[Proposition V.2.8]{CarmonaLacroix1990}.
  Since isolated points of $\Sigma$ are $\P$-a.s.\ eigenvalues of $(H_\omega)$, in view of \cite[Corollary 8.4]{EckhardtTeschl2011} these points belong to $\Sigma_\disc$.
  Hence, also the second assertion holds true.
\end{proof}

\section{Transfer Matrices and Lyapunov exponents}
\label{sec:Transfer}

Let $\omega\in\Omega$ and $z\in\C$.
We say that $u$ is a \emph{(generalized) solution} of the equation $H_\omega u = zu$ if $u\in C(\R)$ and
\[-\Delta u + \omega u = zu\]
in the sense of distributions.

\begin{remark}
\label{rem:solutions}
  Let $u$ be a solution of the equation $H_\omega u = z u$. Then:
  \begin{enumerate}
    \item
      $\Delta u$ is a local measure and hence $u'$ is locally of bounded variation (since its distributional derivative is locally a signed Radon measure), however may not be continuous.
    \item
      As proven in \cite[Lemma 3.4]{SeifertVogt2013}, for all intervals $I\subseteq\R$ of length $1$ we have
      \[\norm{u'|_I}_\infty \leq C_1 \norm{u|_I}_\infty \leq C_2\norm{u|_I}_2\]
      for some constants $C_1,C_2$ only depending on $z$ and $\norm{\omega}_\lu$. Thus, growth conditions on $u$ also imply such conditions on $u'$.
    \item
      $u$ is uniquely determined by $(u(0),u'(0\rlim))$ (see also \cite[Theorem 2.3]{BenAmorRemling2005}). 
      In fact, for any $(a,b)\in\K^2$ there exists a unique solution $u$ such that $(u(0),u'(0\rlim)) = (a,b)$.
  \end{enumerate}
\end{remark}

\begin{remark}
\label{rem:lpc}
  Note that Green's identity (\cite[Theorem 2.2]{BenAmorRemling2005}) holds true for $H_\omega$.
  Hence, one can prove (along the lines of classical Surm-Liouville theory) that we have limit point case for $H_\omega$ at $\pm\infty$, i.e.,
  for $z\in \C^+:=\set{z\in \C; \Im z>0}$ there exist unique (up to some scalar factor) solutions $u_\pm$ of $H_\omega u = zu$
  such that $u_-\in L_2(-\infty,0)$, $u_+\in L_2(0,\infty)$.
\end{remark}

For $t\in \R$ we can define the transfer matrix mapping the solution of $H_\omega u = z u$ at $0$ to the solution at $t$, i.e.,
\[
  T_z(t,\omega)\from \begin{pmatrix} u(0)\\u'(0\rlim)\end{pmatrix} \mapsto \begin{pmatrix} u(t)\\u'(t\rlim)\end{pmatrix}.
\]

Let $\uN,\uD$ be the solutions of $H_\omega u = zu$ satisfying Neumann and Dirichlet boundary conditions at $0$, respectively, i.e.,
\begin{alignat*}{2}
  \uN(0) & = 1,  &  \uD(0) & = 0, \\
  \uN'(0\rlim) & = 0, \qquad & \uD'(0\rlim) & = 1.
\end{alignat*}
Then
\[T_z(t,\omega) =
\begin{pmatrix} \uN(t) & \uD(t) \\ \uN'(t\rlim) &
\uD'(t\rlim)\end{pmatrix}.\]
Furthermore, $\det T_z(t,\omega) = 1$ since the Wronskian of two solutions to the same equation is constant (see \cite[Proposition 2.5]{BenAmorRemling2005}),
and by uniqueness of solutions we obtain $T_z(s+t,\omega) = T_z(t,\alpha_s(\omega))T_z(s,\omega)$ for all $s,t\in\R$.
Thus, $T_z\from\R\times\Omega\to SL(2,\C)$ forms a cocycle.
For $E\in\R$ the cocycle $T_E$ is even $SL(2,\R)$-valued.

\begin{definition}
  We say that $\Omega$ is \emph{atomless}, if $\omega(\set{t}) = 0$ for all $t\in\R$ and $\omega\in\Omega$, i.e., if every $\omega$ in $\Omega$ is a continuous measure.
\end{definition}
Note that if $\Omega$ is atomless if and only if $T_z(t,\cdot)$ is continuous for all $t\in\R$.

Let $(\Omega,\alpha,\P)$ be ergodic and $E\in\R$. By Kingman's subadditive ergodic theorem (see e.g.\ \cite[Corollary IV.1.3]{CarmonaLacroix1990}) there exists $\gamma(E)\in\R$ such that
\[\frac{1}{t} \ln \norm{T_E(t,\omega)}\to \gamma(E) \quad \P\text{-a.s.}\]
$\gamma(E)$ is called the \emph{Lyapunov exponent} for the energy $E$.
We say that $T_E$ is \emph{uniform} if the limit exists for all $\omega\in\Omega$ and the convergence is uniformly in $\Omega$.

\begin{lemma}
\label{lem:trivial_lyapunov_filtration}
  Let $(\Omega,\alpha,\P)$ be uniquely ergodic and atomless, $E\in\R$.
  Then
  \[\limsup_{t\to\infty} \sup_{\omega\in\Omega} \frac{1}{t}
\ln\norm{T_E(t,\omega)}\leq \gamma(E).\]
  Hence, if $\gamma(E) = 0$ then $T_E$ is uniform.
\end{lemma}

\begin{proof}
  Defining $X_t:= \ln\norm{T_E(t,\cdot)}$ the first assertion is a direct consequence of Proposition \ref{prop:subadditive_est}.
  Since $\det T_E(t,\cdot) = 1$ we have $X_t\geq 0$ for all $t$ and therefore $\gamma(E) \geq 0$. Hence, the second assertion is trivial.
\end{proof}

If $T_E$ is uniform and $\gamma(E)>0$ then $T_E$ is sometimes called \emph{uniformly hyperbolic}.
At least if $\Omega$ is atomless we have a characterization of uniform hyperbolicity by means of an exponential splitting.
The proof follows very closely the lines of \cite[Theorem 3]{Lenz2004} for the discrete case (with the obvious modifications for the continuum setting), for details  see also \cite[Theorem 5.2.8]{Seifert2012}. Here, we only state the result.

\begin{theorem}
\label{thm:equivalence_uniformity}
	Let $(\Omega,\alpha,\P)$ be uniquely ergodic and atomless, $E\in\R$.
	Then the following assertions are equivalent:
	\begin{enumerate}
		\item
			$T_E$ is uniform and $\gamma(E)>0$.
		\item
			There exist constants $\kappa,C>0$ and $u,v\in
C(\Omega,\mathcal{P}(\R^2))$ with
			\[\norm{T_E(t,\omega)U}\leq Ce^{-\kappa t} \norm{U}
\quad\text{and}\quad \norm{T_E(-t,\omega)V}\leq Ce^{-\kappa t}
\norm{V}\]
			for all $\omega\in\Omega$, $t\geq 0$, $U\in u(\omega)$ and $V\in
v(\omega)$.
	\end{enumerate}
\end{theorem}

Here $\mathcal{P}(\R^2)$ denotes the projective line, i.e., the set of directions in $\R^2$ (two directions may be identified if they span the same one-dimensional subspace).

\section{Characterization of the Spectrum}
\label{sec:CharSpec}

In this section we characterize the spectrum as a set.

\begin{lemma}
\label{lem:uniform_implies_zeroset_cont}
	Let $(\Omega,\alpha,\P)$ be strictly ergodic and atomless, $T_E$ uniform for every $E$ in
$\R$. Then for the ($\omega$-independent) spectrum we have
	$\Sigma = \set{E\in\R;\; \gamma(E) = 0}$.
\end{lemma}

\begin{proof}
	Set $\mathcal{Z}:=\set{E\in\R;\; \gamma(E) = 0}$.

	Let $\omega\in \Omega$.
	Write
	\begin{align*}
		\mathcal{S} := \big\{E\in \R;\;  & \mbox{ $\forall$ solutions $u$ of $H_\omega u =
Eu$ $\forall$ $\kappa>0$ $\exists$ $C>0$:} \\
		& \abs{u(t)}\leq Ce^{\kappa\abs{t}} \; (t\in\R)\big\},
	\end{align*}
	for the set of energies such that all solutions of the Schr\"odinger equation are subexponentially bounded.

	First of all we show that $\mathcal{Z} \subseteq \mathcal{S}$.
	Let $E\in \mathcal{Z}$. Then
	\[\lim_{t\to \pm \infty} \frac{1}{\abs{t}} \ln \norm{T_E(t,\omega)} = 0.\]
	Hence, for all $\kappa>0$ there is $t_0>0$ such that	
	\[\frac{1}{\abs{t}} \ln \norm{T_E(t,\omega)} \leq \kappa \quad
(\abs{t}>t_0),\]
	i.e. $\norm{T_E(t,\omega)}\leq e^{\kappa\abs{t}}$ for $\abs{t}>t_0$.
	We conclude that $E\in \mathcal{S}$.
%	 There exists
%$C>1$ such that $\norm{T_E(t,\omega)}\leq C$ for $\abs{t}\leq t_0$,
%since solutions remain bounded on compact intervals. This
%implies
%	\[\norm{T_E(t,\omega)}\leq Ce^{\kappa\abs{t}} \quad(t\in \R),\]
%	i.e., $E\in A$.

``$\mathcal{Z}\subseteq \Sigma$'':
	Let $E\in \mathcal{Z}\subseteq \mathcal{S}$ and $u\neq 0$ be a solution of $H_\omega u = Eu$.
	Then $u$ is subexponentially bounded and by Sch'nol-type arguments (see \cite[Theorem 4.4]{BoutetdeMonvelLenzStollmann2009}) we
 conclude that $E \in \sigma(H_\omega) = \Sigma$.
%	By minimality, the spectrum does not depend on $\omega$ and hence
%$\mathcal{Z}\subseteq \Sigma$.

	``$\Sigma\subseteq \mathcal{Z}$'':
	We have to show that $\Sigma = \sigma(H_\omega)\subseteq
 \mathcal{Z}$. We prove this by contradiction. Assume there is spectrum in
$\complement\mathcal{Z}$. By Coppel's Theorem, see e.g.\ \cite[Theorem 3.1]{Johnson1987}, we can deduce that
 $\complement \mathcal{Z}$ is
open and hence the spectral measures of $H_\omega$ give weight to $\complement
\mathcal{Z}$. Therefore, there is $E \in \complement \mathcal{Z}\cap
\sigma(H_\omega)$ admitting a subexponentially bounded solution $u\neq 0$ of
$H_\omega u = Eu$ (see \cite[Theorem 4.6]{BoutetdeMonvelLenzStollmann2009}).
%	We have
%	\[\begin{pmatrix} u(t) \\ u'(t\rlim) \end{pmatrix} = T_E(t,\omega)
%\begin{pmatrix} u(0) \\ u'(0\rlim) \end{pmatrix} \quad(t\in\R).\]

	By Theorem \ref{thm:equivalence_uniformity} there exist $\kappa,C>0$ and $u(\omega),v(\omega)\in \mathcal{P}(\R^2)$ such that
	\[\norm{T_E(t,\omega) U}\leq Ce^{-\kappa t} \norm{U},\quad \norm{T_E(-t,\omega) V}\leq Ce^{-\kappa t} \norm{V}\]
	for all $t\geq 0$, $U\in u(\omega)$, $V\in v(\omega)$,
	and $u(\omega)\neq v(\omega)$. Hence, there exist $U\in u(\omega)$ and
$V\in v(\omega)$ such that
	\[\begin{pmatrix} u(0) \\ u'(0\rlim) \end{pmatrix} = U + V.\]
	Furthermore, for $t\in\R$,
	\[\norm{\begin{pmatrix} u(t) \\ u'(t\rlim) \end{pmatrix}} = \norm{T_E(t,\omega)
\begin{pmatrix} u(0) \\ u'(0\rlim) \end{pmatrix}} \geq
\bigl\lvert{\norm{T_E(t,\omega) U} - \norm{T_E(t,\omega)V}}\bigr\rvert.\]
%	For $t\geq 0$ large, $\norm{T_E(t,\omega) U}$ becomes small, so
%	\[\norm{\begin{pmatrix} u(t) \\ u'(t\rlim) \end{pmatrix}} \geq
%\norm{T_E(t,\omega)V} - \norm{T_E(t,\omega) U}\geq \tilde{C}
%e^{\frac{1}{2}\kappa t}.\]
%	For $-t\geq 0$ large, $\norm{T_E(t,\omega)V}$ becomes small, so
%	\[\norm{\begin{pmatrix} u(t) \\ u'(t\rlim) \end{pmatrix}} \geq
%\norm{T_E(t,\omega)U} - \norm{T_E(t,\omega) V}\geq \tilde{C}
%e^{\frac{1}{2}\kappa t}.\]
	Since the right-hand side is exponentially growing as $\abs{t}\to\infty$,
in view of Remark \ref{rem:solutions}(b) also $u$ is exponentially
growing. This contradicts the fact that $u$ is subexponentially bounded.
\end{proof}

\begin{lemma}
\label{lem:resolvent_implies_positive_and_uniform}
	Let $(\Omega,\alpha)$ be strictly ergodic and atomless, $E\in\R\setminus \Sigma$. Then
$T_E$ is uniformly hyperbolic.
\end{lemma}

\begin{proof}
By minimality, $E\in \rho(H_\omega)$ for all $\omega\in\Omega$.

Let $\omega\in\Omega$.
We show: there exist vectors $U(\omega),V(\omega)\in
\R^2$ such that
$\norm{T_E(t,\omega)U(\omega)}$ decays exponentially for $t\to\infty$ and
$\norm{T_E(t,\omega)V(\omega)}$ decays exponentially for $t\to-\infty$.

%There exists $t_0<0$ such that $\omega(\set{t_0}) = 0$.
Let $t_0<0$.
%Define the restriction $H_\omega|_{[t_0,0]}$ of $H_\omega$ to $[t_0,0]$ by
%\begin{align*}
%D(H_\omega|_{[t_0,0]}) & := \set{u\in L_2(t_0,0);\; u,A_\omega u\in
%W_{1,\loc}^1[t_0,0],\, -(A_\omega u)'\in L_2(t_0,0)},\\
%H_\omega|_{[t_0,0]}u & := -(A_\omega u)'.
%\end{align*}
Since we have limit point case at $-\infty$
there exists $(a,b)\in \R^2\setminus\set{(0,0)}$
such that for solutions $u$ of $H_\omega u= Eu$ with $(u(t_0),u'(t_0\rlim)) \in \lin\set{(a,b)}$
(the linear span of $\set{(a,b)}$)
we have $u\notin L_2(-\infty,t_0)$.
Let $v\in D(H_\omega)$ such that
\[\begin{pmatrix} v(t_0)\\v'(t_0)\end{pmatrix} = \begin{pmatrix}
a\\b\end{pmatrix},
\quad \begin{pmatrix} v(0)\\v'(0)\end{pmatrix} = \begin{pmatrix}
0\\0\end{pmatrix}.\]

%Note, that such a function can be constructed by piecewise $C^1$-functions on $(t_0,0)\setminus\spt \omega_p$.

Set $\tilde{v}:= \1_{(t_0,0)}(H_{\omega} - E)v \in L_2(\R)$ and define
$u:=(H_\omega-E)^{-1}\tilde{v} \in L_2(\R)$. Note that $u$ is a solution of $H_\omega u = Eu + \tilde{v}$ and hence a solution of
$H_\omega u = Eu$ on $[0,\infty)$.
Then $(u(0),u'(0\rlim))\neq (0,0)$, for if $(u(0),u'(0\rlim))=(0,0)$, then $u|_{(t_0,0)} = v|_{(t_0,0)}$ and hence
\[\begin{pmatrix} u(t_0)\\u'(t_0\rlim)\end{pmatrix} = \begin{pmatrix} a\\b\end{pmatrix}.\]
But this would imply $u\notin L_2(-\infty,t_0)$ and therefore $u\notin L_2(\R)$. Hence, $u$ cannot vanish on $[0,\infty)$.

By Combes-Thomas arguments (see e.g.~\cite[Theorem 2.4.1]{Stollmann2001} for a version for forms which also works for measures as potentials)
there exist $C\geq 0$ and $\kappa>0$ (not depending on $\omega$) such that
\[\norm{\1_{(t-\frac{1}{2},t+\frac{1}{2})} u}_{L_2(\R)} \leq Ce^{-\kappa t}
\quad(t\geq 0).\]
By Remark \ref{rem:solutions}(b) then also
\[\norm{\begin{pmatrix} u(t) \\ u'(t\rlim)\end{pmatrix}} \leq \tilde{C}e^{-\kappa t} \quad(t\geq 0)\]
for some $\tilde{C}\geq 0$.

Hence, the initial condition $U(\omega)=(u(0),u'(0\rlim))$ gives rise to a solution of the
Schr\"odinger equation $H_\omega u = Eu$ which decays exponentially for $t\to\infty$ and does not vanish on $[0,\infty)$.
This yields an element
$u(\omega) = [U(\omega)]_{\mathcal{P}(\R^2)} \in \mathcal{P}(\R^2)$.

Analogously, we find $v(\omega)\in \mathcal{P}(\R^2)$ such that the corresponding solutions decay exponentially for $t\to-\infty$.

We have $u(\omega)\neq v(\omega)$. Indeed, in case $u(\omega) = v(\omega)$,
such an initial condition would yield an $L_2(\R)$-solution of $H_\omega u =
Eu$, i.e., $E\notin \sigma(H_\omega)$ would be an eigenvalue of $H_\omega$.

Therefore, $T_E$ admits an exponential splitting (note that the constants $\kappa$ and $C$ can be chosen uniformly on $\Omega$).
By \cite[Lemma 7]{SackerSell1974}, $\omega\mapsto u(\omega)$ and $\omega\mapsto
 v(\omega)$ are continuous.
By Theorem \ref{thm:equivalence_uniformity} we conclude that $T_E$ is uniformly hyperbolic.
\end{proof}

\begin{theorem}
\label{thm:char_spectrum}
	Let $(\Omega,\alpha,\P)$ be strictly ergodic and atomless. Then
	\[\Sigma = \set{E\in\R;\; \gamma(E) = 0} \cup\set{E\in\R;\; T_E \text{ is
not uniform}},\]
	where the union is disjoint.
\end{theorem}

\begin{proof}
	By Lemma \ref{lem:trivial_lyapunov_filtration} the union is disjoint.

	``$\supseteq$'': This is a direct consequence of Lemma
\ref{lem:resolvent_implies_positive_and_uniform}.

	``$\subseteq$'': Let $E\in\R$ with $\gamma(E)>0$ and $T_E$ uniform, $\delta>0$.
	It is easy to see that as soon as $\abs{E-E'}$ is small enough, we have
\[D:=\sup_{-1\leq t\leq 1} \sup_{\omega\in\Omega} \norm{T_E(t,\omega)-T_{E'}(t,\omega)}<\delta.\]
By Coppel's Theorem \cite[Theorem 3.1]{Johnson1987}, $T_{E'}$ is uniformly hyperbolic for all $E'$ in a small open interval $I$ containing $E$.
Now we can repeat the proof of Theorem \ref{lem:uniform_implies_zeroset_cont}
replacing $\complement \mathcal{Z}$ by $I$ to obtain $E\notin \Sigma$.
\end{proof}

\section{Characterization of the absolutely continuous spectrum}
\label{sec:CharacSpec}

For $\mu\in\M$ and $z\in\C^+$ let
\[m_\pm(z,\mu):= \pm \frac{u_\pm'(0\rlim)}{u_\pm(0)},\]
where $u_\pm$ are the unique solutions of $H_\mu u = zu$ being $L_2$ at $\pm\infty$, see Remark \ref{rem:lpc}.
The function $m_\pm(\cdot,\mu)$ are called $m$-functions.

\begin{lemma}[{\cite[Lemma 1]{Remling2007}}]
\label{lem:conv_m-funktion}
  Let $\Omega$ be atomless, $(\omega_n)$ in $\Omega$, $\omega\in\Omega$, $\omega_n\to\omega$ vaguely. Then
  \[m_\pm(\cdot,\omega_n)\to m_\pm(\cdot,\omega)\]
  uniformly on compact subsets of $\C^+$.
\end{lemma}

\begin{lemma}
  Let $\Omega$ be atomess and $K\subseteq \C^+$ compact. Then there exist $C_1,C_2>0$ such that for all $z\in K$ and $\omega\in\Omega$ we have
  \[C_1\leq \Im m_\pm(z,\omega) \leq \abs{m_\pm(z,\omega)} \leq C_2.\]
\end{lemma}

\begin{proof}
  By Lemma \ref{lem:conv_m-funktion} the functions $m_\pm\from K\times\Omega\to\C$ are continuous.
   Since $K\times\Omega$ is compact there exists $C_2\geq 0$ such that
   \[\abs{m_\pm(z,\omega)} \leq C_2\quad(z\in K,\omega\in\Omega).\]
   We show
   \[\min_{z\in K,\omega\in\Omega} \Im m_+(z,\omega)>0\]
   (note that the minimum exists by continuity of $\Im m_+$ and compactness of $K\times\Omega$).
   Indeed, if the minimum was zero there would be $z\in K$, $\omega\in\Omega$ such that $\Im m_+(z,\omega) = 0$.
   By \cite[Theorem 9.1 and Corollary 9.5]{EckhardtTeschl2011} we have
   \[\frac{\Im m_+(z,\omega)}{\Im z} = \norm{u_+}_{L_2(0,\infty)}^2\]
   yielding that $u_+$ is zero on $(0,\infty)$ and hence on $\R$, a contradiction.
\end{proof}

Given the previous two lemmas, we
can extend the Ishii-Pastur-Kotani Theorem to our setting. As the proof follows strictly the  lines of the original paper \cite{Kotani1982}, see also \cite[Section VII.3]{CarmonaLacroix1990},  we only state the result.

Recall that for $A\subseteq \R$ measurable the essential closure of $A$ is defined as
\[\overline{A}^{\ess} := \set{E\in \R;\; \forall\,\varepsilon>0: \lambda(A\cap
(E-\varepsilon,E+\varepsilon))>0},\]
where $\lambda$ denotes Lebesgue measure on $\R$.

\begin{theorem}
\label{thm:Ishii-Pastur-Kotani}
  Let $(\Omega,\alpha,\P)$ be ergodic and atomless. Then
  \[\Sigma_{\ac} = \overline{\set{E\in\R;\; \gamma(E) = 0}}^{\ess}.\]
\end{theorem}

\section{Absence of absolutely continuous spectrum}
\label{sec:AbsacSpec}

In this section we show that a suitable finite local complexity condition on the potential (i.e., the measure) combined with aperiodicity yields absence of absolutely continuous spectrum.
First, we recall some definitions from \cite{KlassertLenzStollmann2011}.

\begin{definition}
  A \emph{piece} is a pair $(\nu,I)$ consisting of an interval $I\subseteq\R$ with positive length
  $\lambda(I) >0$ (which is then called the \emph{length} of the piece) and a local measure $\nu$ on $\R$ supported on $I$.
  We abbreviate pieces by $\nu^I$. Without restriction, we may assume that $\min I = 0$.
  A \emph{finite piece} is a piece of finite length.
  We say $\nu^I$ \emph{occurs} in a local measure $\mu$ at $x\in\R$, if $\1_{x+I}\mu$
  is a translate of $\nu$.

	The \emph{concatenation} $\nu^I=\nu_1^{I_1}\mid \nu_2^{I_2}\mid \ldots$ of a finite or countable family
    $(\nu_j^{I_j})_{j\in N}$, with $N=\set{1,2,\ldots,\abs{N}}$ (for $N$
finite) or $N=\N$ (for $N$ infinite), of finite pieces is defined by
	\begin{align*}
		I & = \left[\min I_1,\min I_1 + \sum_{j\in N} \lambda(I_j)\right],\\
		\nu & = \nu_1+\sum_{j\in N,\,j\geq 2} \nu_j\Big(\cdot-\Big(\min I_1 + \sum_{k=1}^{j-1} \lambda(I_k) - \min I_j\Big)\Big).
	\end{align*}
	We also say that $\nu^I$ is \emph{decomposed} by $(\nu_j^{I_j})_{j\in N}$.
\end{definition}

\begin{definition}
	Let $\mu$ be a local measure on $\R$. We say that $\mu$ has the \emph{finite decomposition property} (f.d.p.),
	if there exist a finite set $\mathcal{P}$ of finite pieces (called the \emph{local pieces})
	and $x_0\in\R$, such that $(\1_{[x_0,\infty)}\mu)^{[x_0,\infty)}$ is a translate of a
	concatenation $v_1^{I_1}\mid \nu_2^{I_2}\mid\ldots$ with $\nu_j^{I_j}\in\mathcal{P}$ for all $j\in\N$.

	A local measure $\mu$ has the \emph{simple finite decomposition property} (s.f.d.p.), if it has the f.d.p.~with a decomposition such that there is $\ell>0$ with the following property: Assume that the two pieces
	\[\nu_{-m}^{I_{-m}} \mid \ldots \mid \nu_{0}^{I_{0}} \mid \nu_{1}^{I_{1}} \mid \ldots \mid \nu_{m_1}^{I_{m_1}} \quad \text{and} \quad
	 \nu_{-m}^{I_{-m}} \mid \ldots \mid \nu_{0}^{I_{0}} \mid \mu_{1}^{J_{1}} \mid \ldots \mid \mu_{m_2}^{J_{m_2}}\]
	occur in the decomposition of $\mu$ with a common first part $\nu_{-m}^{I_{-m}} \mid \ldots \mid \nu_{0}^{I_{0}}$ of length at least $\ell$ and such that
	\[\1_{[0,\ell)}(\nu_{1}^{I_{1}} \mid \ldots \mid \nu_{m_1}^{I_{m_1}}) = \1_{[0,\ell)}(\mu_{1}^{J_{1}} \mid \ldots \mid \mu_{m_2}^{J_{m_2}}),\]
	where $\nu_j^{I_j}$, $\mu_k^{J_k}$ are pieces from the decomposition (in particular, all belong to $\mathcal{P}$ and start at $0$) and the latter two concatenations are of lengths at least $\ell$. Then
	\[\nu_1^{I_1} = \mu_1^{J_1}.\]
\end{definition}

\begin{theorem}
\label{thm:H_mu_abs_spectrum}
  Let $\mu\in\M$ such that $\mu$ and the reflected measure $\mu(-(\cdot))$ have the s.f.d.p.~and assume that neither
  $\mu$ nor $\mu(-(\cdot))$ are eventually periodic.
  Then $H_\mu$ does not have absolutely continuous spectrum.
\end{theorem}

\begin{proof}
  By \cite[Theorem 4.1]{KlassertLenzStollmann2011}, the halfline operators $H_\mu|_{[0,\infty)}$ and $H_{\mu(-(\cdot))}|_{[0,\infty)}$ (with Dirichlet boundary conditions at $0$)
  do not have absolutely continuous spectrum. Let $U\from L_2((-\infty,0])\to L_2([0,\infty))$, $Uf(t):= f(-t)$. Then $U$ is unitary and
  \[U^*H_{\mu(-(\cdot))}|_{[0,\infty)}U = H_\mu|_{(-\infty,0]}.\]
  Hence, both operators $H_\mu|_{[0,\infty)}$ and $H_\mu|_{(-\infty,0]}$ do not have absolutely continuous spectrum.
  Therefore, also $H_\mu$ cannot have any absolutely continuous spectrum.
\end{proof}

We can now state the main theorem of this section. A similar result is stated in \cite[Theorem 5.1]{KlassertLenzStollmann2011}.

\begin{theorem}
\label{thm:Sigma_ac_empty}
  Let $(\Omega,\alpha,\P)$ be ergodic, minimal, aperiodic (i.e.\ there exists $\omega\in\Omega$ which is not periodic) and
  have the
    s.f.d.p.\ (i.e. for every $\omega\in\Omega$: $\omega$ and $\omega(-(\cdot))$ have s.f.d.p.).
  Then $\Sigma_{\ac} = \varnothing$.
\end{theorem}
%\Hmm{Ist ergodic noetig? Naja, hier wird ja etwas \"uber $\Sigma_\ac$, also das fast sichere ac-Spektrum ausgesagt. N\"otig ist das nicht, dann ist die Aussage aber, dass die Menge der $\omega$'s mit ac-Spektrum Whk. Null hat. Braucht man wirklich sfdp in beide Richtungen? Wenn man Last/Simon f\"ur unseren Fall hat, dann nicht (siehe dort den Beweis von Thm. 6.1, dass die Halbachsenoperatoren gleiches ac-Spektrum haben). In der Arbeit gibt es auch die Bemerkung, dass alle minimal ergodischen Prozesse das Resultat liefern. In der Arbeit wird aber auch irgendwo mal die Beschr\"anktheit des Potentials benutzt. Moralisch sollte das auch f\"ur Ma{\ss}e gelten, das sollte sicher auch mal sauber aufgeschrieben werden.}

\begin{proof}
  Assume that $\set{\omega\in\Omega;\; \sigma_{\ac}(H_\omega)\neq\varnothing}$ has positive $\P$-measure.
  By Theorem \ref{thm:H_mu_abs_spectrum} the set $\set{\omega\in\Omega;\; \text{$\omega$ or $\omega(-(\cdot))$ is
  eventually periodic}}$
  has positive $\P$-measure. W.l.o.g.~assume that $\omega$ is periodic for $t\geq t_0$ with period $p$.
  By closedness of $\Omega$,
  \[\tilde{\omega}:= \lim_{t\to\infty} \alpha_{t}(\omega) = \lim_{t\to\infty} \omega(\cdot+t)\in \Omega\]
  and $\tilde{\omega}$ is periodic with period $p$.
  For $\omega'\in\Omega$ there exists $(t_n)$ in $\R$ such that
  $\alpha_{t_n}(\tilde{\omega})\to \omega'$. Since $\tilde{\omega}$ is $p$-periodic and $\alpha$ is continuous, we arrive at
  \[\alpha_p(\omega') = \alpha_p\left(\lim_{n\to\infty} \alpha_{t_n}(\tilde{\omega})\right) = \lim_{n\to\infty} \alpha_{t_n}\alpha_p(\tilde{\omega}) = \omega'.\]
  So, every $\omega\in\Omega$ must be periodic with the same period, a contradiction.
\end{proof}

For applications it might be helpful to have some condition when a measure
actually has the s.f.d.p. The remaining part of this section will address this
issue.

\begin{proposition}[{\cite[Proposition 2]{KlassertLenzStollmann2011}}]
\label{prop:fdp_sfdp}
  Let $\mu\in\M$ have the f.d.p.\ with local pieces $\nu_1^{I_1},\ldots,\nu_N^{I_N}$. If $\mu$ does not have the s.f.d.p.~with respect to any set of local pieces,
  then there must be two local pieces among the $\nu_j^{I_j}$ which are multiples of Lebesgue measure, i.e. of the form $c\1_I \lambda$.
\end{proposition}

\begin{lemma}
\label{lem:sfdp_contruction}
Let $A$ be a finite set, $D:=\set{x_n;\;n\in\Z}\subseteq \R$ be of finite local complexity, i.e.\ $D':=\set{x_{n+1}-x_n;\;n\in\Z}$ is finite.
Let $f\from D\to A$ and for $a\in A$ let $\nu_a\in \M$ be supported in $[0,l_a]$. Assume that $l_{f(x_n)} \leq x_{n+1}-x_n$ for all $n\in\Z$. Define
\[\mu:= \sum_{x\in D} \delta_x*\nu_{f(x)} = \sum_{x\in D} \nu_{f(x)}(\cdot - x).\]
Then $\mu$ and $\mu(-(\cdot))$ have the f.d.p. If, additionally, at most one of the $\nu_a$ is a multiple of Lebesgue measure then $\mu$ and $\mu(-(\cdot))$ have the s.f.d.p.
\end{lemma}

\begin{proof}
  By construction the whole measure $\mu$ can be decomposed into the pieces $\set{\mu_a^{[0,l_a]};\; a\in A}$; therefore $\mu$ has the f.d.p. The second assertion is a direct consequence of Proposition \ref{prop:fdp_sfdp}.
\end{proof}

\section{Cantor Spectra of zero measure}
\label{sec:Cantor}

We now prove Cantor spectra for a large class of operators in case of atomless $\Omega$.
We call $C\subseteq \R$ a \emph{Cantor set} if $C$ is closed, nowhere dense (i.e.\ $C$ does not contain any interval of positive length) and does not contain any isolated points.

\begin{theorem}
\label{thm:cantor1}
  Let $(\Omega,\alpha,\P)$ be strictly ergodic, atomless, aperiodic and have the s.f.d.p.
  Furthermore, let $T_E$ be uniform for all $E\in\R$. Then
  \[\Sigma = \set{E\in\R;\; \gamma(E) = 0}\]
  and $\Sigma$ is a Cantor set of zero Lebesgue measure.
\end{theorem}

\begin{proof}
  By Theorem \ref{thm:char_spectrum} we observe $\Sigma = \set{E\in\R;\; \gamma(E) = 0}$. By Theorem \ref{thm:Ishii-Pastur-Kotani} we have
  \[\overline{\set{E\in\R;\; \gamma(E) = 0}}^{\ess} = \Sigma_{\ac}.\]
  Since $\Sigma_{\ac} = \varnothing$ by Theorem \ref{thm:Sigma_ac_empty} we
infer $\lambda(\set{E\in\R;\; \gamma(E) = 0}) = 0$.
  Note that $\Sigma$ does not contain any isolated
points by Lemma \ref{lem:disc_spectrum}.
%
%  We show that $\Sigma$ does not contain isolated points.
%  Indeed, assume that $E\in\Sigma$ was such an isolated point in the spectrum.
%
%  Then $E$ would be an eigenvalue of $H_\omega$ for $\P$-a.a.
%$\omega\in\Omega$.
%  By \cite[Corollary 8.4]{EckhardtTeschl2011}, the multiplicity of $E$ would
%be $1$, so $E$ belongs to the discrete spectrum
%  (the isolated eigenvalues of finite multiplicity) of $H_\omega$ for
%$\P$-a.a.~$\omega\in\Omega$.
%  By ergodicity there exists $\Sigma_{\disc}\subseteq \R$ such that
%$\Sigma_{\disc}$ is the discrete spectrum of $H_\omega$ for
%$\P$-a.a.~$\omega\in\Omega$,
%  see \cite[Remark V.2.5]{CarmonaLacroix1990}.
%  But $\Sigma_{\disc} = \varnothing$ due to \cite[Proposition
%V.2.8]{CarmonaLacroix1990}.
%
%  So, $\Sigma$ is closed and every point in $\Sigma$ is a limit point of
%$\Sigma$.
  Since $\lambda(\Sigma) = 0$ we conclude that $\Sigma$ is nowhere dense, i.e., Cantor set of zero Lebesgue measure.
\end{proof}

It remains to establish some criterion for uniformity of the transfer matrices.

Let $A$ be a finite set equipped with the discrete topology.
A pair $(X,S)$ is a \emph{subshift} over $A$ if $X$ is a closed subset of
$A^\Z$,
where $A^\Z$ is endowed with the product topology, and $X$ is invariant under
the shift $S\from A^\Z\to A^\Z$, $S a(n) := a(n+1)$.

For $a\in A$ let $\nu_a\in \M$ be atomless and supported on $[0,l_a]$.
%Without loss of generality let $\inf\spt\nu_a = 0$ for all $a\in A$.

%Let $\nu_1,\ldots,\nu_N\in \M$ with compact support such that
%$\inf\spt\nu_j=0$ for all $j\in\set{1,\ldots,N}$. For $j\in\set{1,\ldots,N}$
%we
% define
%\[l_j:=\begin{cases}
%	\sup\spt\nu_j & \text{if}\;\spt\nu_j\not\subseteq \set{0},\\
%	1 & \text{if}\; \spt\nu_j\subseteq \set{0}.
%      \end{cases}\]
For $x\in X$ we define the measure $\omega_x\in \M$ by
\[\omega_x:= \sum_{n\in\N_0} \delta_{\sum_{k=0}^{n-1} l_{x(k)}} * \nu_{x(n)} + \sum_{n\in\N} \delta_{\sum_{k=-n}^{-1} -l_{x(k)}} * \nu_{x(-n)}.\]
Let
\[\Omega:= \set{\alpha_{t}(\omega_x);\; x\in X,\, t\in\R}.\]

Many properties of $(X,S)$ transfer to $(\Omega,\alpha)$, as the following
proposition shows.

\begin{proposition}[{\cite[Proposition 4]{KlassertLenzStollmann2011}}]
\label{prop:subshift}
  \begin{enumerate}
   \item
      Assume that at most one of the measures $\nu_a$ is a multiple of Lebesgue measure.
      Then $(\Omega,\alpha)$ has the s.f.d.p.
   \item
      Any invariant probability measure $\P_X$ on $(X,S)$ induces a canonical
invariant probability measure
      $\P$ on $(\Omega,\alpha)$.
      If $\P_X$ is ergodic, then $\P$ is ergodic.
    \item
      If $(X,S)$ is uniquely ergodic, then $(\Omega,\alpha)$ is uniquely
ergodic.
    \item
      If $(X,S)$ is minimal, then $(\Omega,\alpha)$ is minimal.
  \end{enumerate}
\end{proposition}

\begin{proof}
  Part (a) is just Lemma \ref{lem:sfdp_contruction}. For parts (b) to (d) note that $(\Omega,\alpha^\Z)$ is a factor of $(X,S)$ where $\alpha^\Z$ is the restriction of $\alpha$ to $\Z\times\Omega$.
  Thus, all the assertions hold true for $(\Omega,\alpha^\Z)$ instead of $(\Omega,\alpha)$, cf.\ \cite{BaakeLenz2005}. This, however, implies the assertions also for $(\Omega,\alpha)$.
\end{proof}

The crucial notion for uniformity will be Boshernitzan's condition introduced
in \cite{Boshernitzan1985}.

\begin{definition}
  Let $(X,S)$ be a subshift over a finite alphabet $A$. Then $(X,S)$ satisfies
condition (B) if there exists an ergodic probability measure $\P_X$ on $X$
  with
  \[\limsup_{n\to\infty} n\eta_{\P_X}(n) >0,\]
  where $\eta_{\P_X}(n):=\min\set{\P_X(V_w);\; w\in\mathcal{W},\, \abs{w} = n}$, $\mathcal{W}$ is the set of finite words and $V_w := \set{x\in X;\; x(1)\cdots x(\abs{w}) = w}$ is the cylinder set to $w$.
\end{definition}
Condition (B) has been introduced by Boshernitzan as a sufficient condition for unique ergodicity. It was then shown to imply a strong form of subadditive ergodic theorem in \cite{DamanikLenz2006}. This was used there in connection with the method of \cite{Lenz2002} to  establish Cantor spectrum of Lebesgue measure zero for aperiodic subshift satisfiying (B). Our use of this condition below is in exactly the same spirit. The  condition can be seen to hold
for a large number of subshifts, see e.g.~\cite{DamanikLenz06b}.

\begin{theorem}
\label{thm:Cantor2}
  Let $(X,S)$ be a minimal subshift over $A$ satisfying (B) and $\nu_a\in\M$
atomless for all $a\in A$. Then $T_E$ is uniform for all $E\in\R$.
\end{theorem}

\begin{proof}
  Let $E\in\R$, $x\in X$. For $n\in\N_0$ let $s_n:= l_{x(0)} + \ldots + l_{x(n-1)}$. Consider
  \[T^{\disc}_E(n,\omega_x):= T_E(s_n, \omega_x) \quad(n\geq 0).\]
  Then $T^{\disc}_E(m+n,\omega_x) = T^{\disc}_E(m,\omega_{S^n
x})T^{\disc}_E(n,\omega_{x})$.
  Thus $T^{\disc}_E$ is a continuous $SL(2,\R)$-valued cocycle (with discrete time). By \cite[Theorem 1]{DamanikLenz2006}, $T^{\disc}_E$ is uniform.
  Let $t\in\R$, $s>\max\set{-t,0}$ and choose $k\in\N_0$ such that $s_k\leq s < s_{k+1}$. Then
  \begin{align*}
    & \frac{1}{s} \ln\norm{T_E(s,\alpha_t(\omega_x))} - \frac{1}{s_k} \ln \norm{T_E(s_k,\omega_x)} \\
    & \leq \frac{1}{s} \ln\norm{T_E(s+t,\omega_x)} + \frac{1}{s}\ln\norm{T_E(t,\omega_x)^{-1}} - \frac{1}{s_k} \ln \norm{T_E(s_k,\omega_x)} \\
    & \leq \frac{1}{s_k}\ln \norm{T_E(s+t-s_k,\alpha_{s_k}(\omega_x))} + \frac{1}{s}\ln\norm{T_E(t,\omega_x)^{-1}},
  \end{align*}
  and similarly
  \begin{align*}
    & \frac{1}{s_k} \ln \norm{T_E(s_k,\omega_x)} - \frac{1}{s} \ln\norm{T_E(s,\alpha_t(\omega_x))} \\
 %   & \leq \frac{1}{s} \ln \norm{T_E(s_k,\omega_x)} + \bigl(\frac{1}{s_k}-\frac{1}{s}\bigr) \ln \norm{T_E(s_k,\omega_x)} - \frac{1}{s} \ln\norm{T_E(s+t,\omega_x)} + \frac{1}{s} \ln \norm{T_E(t,\omega_x)} \\
    & \leq \frac{1}{s} \ln \norm{T_E(s_k\!-\!s\!-\!t,\alpha_{s+t}(\omega_x))} + \frac{s-s_k}{ss_k} \ln \norm{T_E(s_k,\omega_x)} + \frac{1}{s} \ln \norm{T_E(t,\omega_x)}.
  \end{align*}
  As $s\to \infty$, the right-hand sides converge to zero, uniformly in $t$ and
 $x$ (note that it suffices to consider $\abs{t}\leq \max_a l_a$ by
$S$-invariance). Thus, also $T_E$ is uniform.
\end{proof}

\begin{theorem}
  Let $(X,S)$ be an aperiodic minimal subshift over $A$ satisfying (B),
  $\nu_a\in \M$ compactly supported, atomless ($a\in A$) and assume that at
  most one of the measures $\nu_a$ is a multiple of Lebesgue measure. Furthermore, assume that $\Omega$ is aperiodic. Then
  \[\Sigma = \set{E\in\R;\; \gamma(E) = 0}\]
  is a Cantor set of zero Lebesgue measure.
\end{theorem}

\begin{proof}
  By Theorem \ref{thm:Cantor2}, $T_E$ is uniform for all $E\in\R$, and in view of Proposition \ref{prop:subshift}(a) we can apply Theorem \ref{thm:cantor1} to obtain the assertion.
\end{proof}

\begin{remark}
  We remark that aperiodicity of $(X,S)$ may not imply aperiodicity of $(\Omega,\alpha)$. Indeed, just choose $\nu\in\M$ with support in $[0,1]$ and let $\nu_a:=\nu$ and $l_a:=1$ for all $a\in A$.
  Then $(\Omega,\alpha)$ is periodic irregardless of $(X,S)$.
\end{remark}

%\Hmm{Kommentar zu obiger Bemerkung, wann Aperiodizit\"at von $(\Omega,\alpha)$ gefolgert werden kann?}

\appendix

\section{(Semi-)Uniform ergodic theorems}

We state and prove a (semi-)uniform ergodic theorem for continuous time.  This theorem is used in the considerations of Section \ref{sec:Transfer}. We adapt
 the
 arguments of \cite[Theorem 1.5]{SturmanStark2000} from the time-discrete case.

\begin{proposition}
\label{prop:subadditive_est}
Let $(\Omega,\alpha,\P)$ be uniquely ergodic and $(X_t)_{t\geq0}$ be a
continuous subadditive process on $\Omega$, i.e.
\[X_0 = 0,\quad X_{t+s} \leq X_t + X_s\circ \alpha_t \quad(s,t\geq0),\]
and $X_t\in C(\Omega)$ for $t\geq 0$.
Furthermore, assume that
\[M:=\sup_{t\in[0,1]}\sup_{\omega\in\Omega} \abs{X_t(\omega)}<\infty.\]
Then there exists $\overline{X}\in \R$ such that $\frac{1}{t}X_t\to \overline{X}$ $\P$-a.s., and we
have
\[\limsup_{t\to\infty} \sup_{\omega\in\Omega} \frac{1}{t} X_t(\omega) \leq \overline{X}.\]
\end{proposition}

\begin{proof}
  Let $\varepsilon>0$. For $t>0$ define $\overline{X}_t:= \frac{1}{t}\int X_t\,
 d\P$. By Kingman's subadditive ergodic theorem
  (see, e.g., \cite[Corollary IV.1.3]{CarmonaLacroix1990}) there exists
$\overline{X}\in\R$ such that $\overline{X}_t\to \overline{X}$; there exists $S\in \N$ such that
$\overline{X}_t\leq \overline{X}+\varepsilon$ for
$t\geq S$.
  Let $K:=\sup_{t\in[0,S]}\sup_{\omega\in\Omega} \abs{X_t(\omega)}< \infty$
(which is finite by subadditivity).

  Let $\omega\in\Omega$.
  By subadditivity, for $k\in\N$ and $t\in[0,S]$ we have
  \[X_{kS}(\omega)\leq X_t(\omega) + \sum_{j=0}^{k-2}
X_S(\alpha_{jS+t}(\omega)) + X_{S-t}(\alpha_{(k-1)S+t}(\omega)).\]
  Integrating with respect to $t$ and dividing by $S$ yields
  \[X_{kS}(\omega) \leq 2K + \sum_{j=0}^{k-2} \frac{1}{S}\int_0^S
X_S(\alpha_{jS+t}(\omega))\, dt = 2K + \frac{1}{S} \int_0^{(k-1)S}
X_S(\alpha_t(\omega))\, dt.\]

  Since $X_S$ is continuous there exists $S'>0$ not depending on $\omega$ such
that for all $t\geq S'$ we have
  \[\frac{1}{t}\int_0^t \frac{1}{S} X_S(\alpha_r(\omega))\, dr \leq \int_\Omega
 \frac{1}{S} X_S(\omega)\, d\P(\omega) + \varepsilon.\]
  Choose $k\in\N$ such that $(k-1)S>S'$. Then
  \begin{align*}
    X_{kS}(\omega) & \leq 2K  + (k-1)S \frac{1}{(k-1)S} \int_0^{(k-1)S}
\frac{1}{S}  X_S(\alpha_t(\omega))\, dt \\
    & \leq  2K  + (k-1)S \overline{X}_S + (k-1)S\varepsilon.
  \end{align*}
  Now, for $t\geq S'+2S$ write $t=kS+r$ with $k\in\N$ and $0\leq r<S$. Then
$(k-1)S = t-r-S>S'$ and therefore
  \[X_t(\omega)\leq X_{kS}(\omega) + X_r(\alpha_{kS}(\omega)) \leq 3K  + (k-1)S
 \overline{X}_S + (k-1)S\varepsilon.\]
  Since $t>(k-1)S$ we obtain
  \[\frac{1}{t}X_t(\omega) \leq \overline{X}_S + \varepsilon + \frac{3K}{t}
\leq \overline{X}+2\varepsilon + \frac{3K}{t}.\]
  For $t\geq T:=\max\set{3K\varepsilon^{-1}, S'+2S}$ we finally arrive at
  \[\frac{1}{t}X_t(\omega) \leq \overline{X}+3\varepsilon.\]
  Thus,
  \[\sup_{\omega\in\Omega} \frac{1}{t}X_t(\omega) \leq \overline{X}+3\varepsilon
\quad(t\geq T).\]
  So,
  \[\limsup_{t\to\infty}\sup_{\omega\in\Omega} \frac{1}{t}X_t(\omega) \leq
\overline{X}+3\varepsilon.\]
  For $\varepsilon\to 0$ we obtain the assertion.
\end{proof}

In a similar way we get uniform control of a lower bound in case of additive
processes. Hence, we can obtain uniform convergence in that case.

\begin{proposition}
\label{thm:additive_conv}
Let $(\Omega,\alpha,\P)$ be uniquely ergodic and $(X_t)_{t\geq0}$ be a
continuous additive process on $\Omega$, i.e.,
\[X_0 = 0,\quad X_{t+s} = X_t + X_s\circ \alpha_t \quad(s,t\geq0),\]
and $X_t\in C(\Omega)$ for $t\geq 0$.
Furthermore, assume that
\[M:=\sup_{t\in[0,1]}\sup_{\omega\in\Omega} \abs{X_t(\omega)}<\infty.\]
Then there exists $\overline{X}\in \R$ such that $\frac{1}{t}X_t\to \overline{X}$ $\P$-a.s., and
\[\lim_{t\to\infty} \sup_{\omega\in\Omega}\abs{ \frac{1}{t} X_t(\omega) - \overline{X}} =
0.\]
\end{proposition}

\bigskip

\noindent
Daniel Lenz \\
Friedrich-Schiller-Universit\"at Jena \\
Fakult\"at f\"ur Mathematik \\
Ernst-Abbe-Platz 2 \\
07743 Jena, Germany \\
{\tt daniel.lenz@uni-jena.de}

\bigskip

\noindent
Christian Seifert \\
Technische Universit\"at Hamburg-Harburg \\
Institut f\"ur Mathematik \\
21073 Hamburg, Germany \\
{\tt christian.seifert@tuhh.de}

\bigskip

\noindent
Peter Stollmann \\
Technische Universit\"at Chemnitz \\
Fakult\"at f\"ur Mathematik \\
09107 Chemnitz, Germany \\
{\tt P.Stollmann@mathematik.tu-chemnitz.de}

\end{document}